\documentclass[10pt, letterpaper, twocolumn]{IEEEtran} 

\usepackage{amsmath}
\usepackage{amssymb}    
\usepackage{amsfonts}
\usepackage[english]{babel}
\usepackage{cite} 
\usepackage{multirow}
\usepackage{stfloats}    

\usepackage{algorithm}   
\usepackage{algorithmic} 

\usepackage{caption} 
\usepackage{graphicx}
\usepackage{epsfig}
\usepackage{subfigure} 

\usepackage{tablefootnote}
\usepackage{diagbox}   
\usepackage{multirow}\usepackage{makecell}
\usepackage{setspace} 
\newcommand{\tabincell}[2]{\begin{tabular}{@{}#1@{}}#2\end{tabular}}
\usepackage{threeparttable}

\usepackage{url}

\usepackage{accents}
\makeatletter
\def\widebar{\accentset{{\cc@style\underline{\mskip10mu}}}}
\def\Widebar{\accentset{{\cc@style\underline{\mskip13mu}}}}
\makeatother

\usepackage{footnote}
\usepackage{pifont}
\usepackage[perpage]{footmisc}  

\newtheorem{theorem}{Theorem}

\newtheorem{definition}{Definition}
\newtheorem{proposition}{Proposition}
\newtheorem{corollary}{Corollary}

\usepackage{booktabs}    
\allowdisplaybreaks  

 \usepackage[draft]{changes}   

\begin{document}
\captionsetup[figure]{name={Fig.    },labelsep=period}  

\title{Age-optimal Service and Decision Scheduling in Internet of Things}

\author{Zhiwei~Bao,~Yunquan~Dong,~\IEEEmembership{Member,~IEEE},~Zhengchuan~Chen,~\IEEEmembership{Member,~IEEE},
~Pingyi~Fan,~\IEEEmembership{Senior~Member,~IEEE},~and~Khaled~Ben~Letaief,~\IEEEmembership{Fellow,~IEEE}
\\

\thanks{Z. Bao and Y. Dong are with the School of Electronic and Information Engineering,
 Nanjing University of Information Science and Technology, Nanjing, 210044, China (e-mail: zwbao@nuist.edu.cn and yunquandong@nuist.edu.cn).

Z. Chen is with College of Microelectronics and Communication Engineering, Chongqing University, Chongqing 400044, China (e-mail: czc@cqu.edu.cn).

P. Fan is with the Department of Electronic Engineering, the Beijing National Research Center for Information Science and Technology, Tsinghua, University, Beijing 100084, China. (e-mail: fpy@tsinghua.edu.cn).

Khaled B. Letaief is with the Department of Electrical and Computer Engineering, HKUST, Clear Water Bay, Kowloon, Hong Kong (e-mail: eekhaled@ust.hk).}
 }

\maketitle
\thispagestyle{empty}
\pagestyle{empty}

\begin{abstract}
We consider an Internet of Things (IoT) system in which a sensor observes a  phenomena of interest with exponentially distributed intervals and delivers the updates to a monitor with the First-come-First-served (FCFS) policy.
    At the monitor, the received updates are used to make decisions with deterministic or random intervals.
        For this system, we investigate the freshness of the updates at these decision epochs using the age upon decisions (AuD) metric.
Theoretical results show that 1) when the decisions are made with exponentially distributed intervals, the average AuD of the system is smaller if the service time (e.g., transmission time) is uniformly distributed than when it is exponentially distributed, and would be the smallest if it is deterministic;
    2)when the decisions are made periodically, the average AuD of the system is larger than, and decreases with decision rate to, the average AuD of the corresponding system with Poisson decision intervals;
3)the probability of missing to use a received update for any decisions is decreasing with the decision rate, and is the smallest if the service time is deterministic.
    For IoT monitoring systems, therefore, it is suggested to use deterministic monitoring schemes, deterministic transmitting schemes, and Poisson decision schemes, so that the received updates are as fresh as possible at the time they are used to make decisions.
\end{abstract}

\IEEEpeerreviewmaketitle

\section{Introduction}
The Internet of Things (IoT) technology has been developed rapidly and has been applied in various scenarios in recent years.
    By connecting smart devices (e.g., sensors, radio frequency identification nodes, infrared sensors,  global positioning system based locators) to the internet, IoT enables effective  information exchange and communication among them, which prompts a modern network of intelligence.
In particular, IoT has spawned more and more real-time applications, e.g., smart transportation, health monitoring, intelligent agriculture, smart home, environment monitoring, and so on~\cite{Nelson-2019,Ahmed-2018,Habibzadeh-2020}.
    For these applications, timely information updates are required to control traffic, or to monitor patient, or to manage the farmland.
Therefore, the timeliness of received updates is critical.

It is noted that, however, traditional measures like delay and throughput are not suitable to characterize the timeliness of updates.
    For example, when the delay is small, the received updates may not be fresh if the updates arrive very infrequently; when the throughput is large, the received updates may also not be fresh due to the crowded update queue at the transmitter.
To this end, a new metric termed as age of information (AoI) was proposed in 2011  in \cite{kaul-2011}.
    Specifically, AoI is defined as the elapsed time since the generation of the latest received update, and thus can characterize the freshness of the received information exactly.

 By its definition, it is seen that AoI naturally records the freshness of the latest available information of the receiver for every moment.
    In many applications, however, we may only be interested in the freshness of the received information at some desired instants.
For example, in the dynamic route planning of vehicular networks, the freshness of collected positioning information of cars is mostly concerned at epochs when we are approaching road crosses, i.e., the moments when these information are used to reschedule the route.
    In view of this, the metric age upon decisions (AuD) was proposed to evaluate the freshness of received information at these \textit{decision epochs} \cite{Dong-2018}.
With this special focus on the decision epochs, AuD is thus very suitable for most IoT systems, which are not only platforms of information collecting and exchanging, but also are unattended decision making eco-networks.

\subsection{Motivation}
In the AuD framework, there are three dimensions of optimization, i.e., the arrival process, the service process, and the decision process.
    For example, when the inter-arrival time between neighboring updates is generally distributed, the service time of each update is exponentially distributed, and the decision process is a Poisson process, we refer to the system as a \textit{G/M/1/M update-and-decide system}.
As shown in\cite{Dong-2019}, the average AuDs of such systems are independent of the decision rate and are minimized by the periodic arrival process (i.e., the deterministic process).
    However, how the service process affects the timeliness of the system decisions has not been well understood yet.
In this paper, therefore, we are interested in the timeliness of the update-and-decide systems with different service process.
In particular, we shall investigate the average AuD of M/G/1/M update-and-decide systems with exponentially distributed inter-arrival times, generally distributed service times, and Poisson decisions.
     We shall also investigate the average AuD of M/G/1/D update-and-decide systems, in which decisions are made periodically.
In doing so, we shall provide answer to the question what service process and decision process will optimize the IoT-based update-and-decide systems.

\subsection{Main Contributions}
In this paper, we assume that the arrivals of updates follow a Poisson process and consider three typical service processes, in which the service times are uniformly distributed, exponentially distributed, and deterministic, respectively.
\begin{itemize}

\item For a system with Poisson decisions, we explicitly present the average AuD of the system when the general service process, the Poisson service process, the uniform service process, or the periodic service is used, respectively.
        We also show that the periodic (i.e., deterministic) service process performs better than the uniform service process, and the exponential service process performs the worst, i.e., achieves the maximum average AuD.
\item For a system using the periodic decision process, we also explicitly present the average AuD of the system when the Poisson service process, the uniform service process, or the periodic service is used, respectively.
        For this system, we show that the periodic service process also performs the best and the exponential service process performs the worst, as in the system with Poisson decisions.
    Moreover, we show that is the average AuD is larger than and will decrease with decision rate to the average AuD of the corresponding system with Poisson decisions.
\item We explicitly present the probability for the receiver missing to use a received update for any decision in order to qualify the utilization of the received updates.
        We show that the missing probability is decreasing with the increase in decision rate and would be the smallest if the service process and the decision process are both periodic.
\end{itemize}

\subsection{Organizations}
The rest of the paper is organized as follows.
In Section~\ref{sec:model}, we present the update-and-decide system model and the definition of AuD.
In Section~\ref{sec:mg1}, we investigate the average AuD and the missing probability of  M/G/1/M systems and present the obtained results via numerical and Monte Carlo results, in which Poisson decision processes are used.
In Section~\ref{sec:MG1D}, we investigate the average AuD and the missing probability of  M/G/1/D systems in which periodic decision processes are used.
Finally, we conclude our work in Section~\ref{sec:conclusion}.

\subsection{Related Works}
The age of information has been exhaustively studied in various queueing systems with First-come-First-served (FCFS) discipline, e.g., the M/M/1 system, the M/D/1 system, and the D/M/1 system \cite{skkaul-2012, kam-2013, Skaul-2012}.
    Under the Last-come-First-served (LCFS) service discipline,  M/M/1 queues with and without service preemption were also analyzed in \cite{Yates-2012}.
Due to its specialty in qualifying the freshness of received updates, AoI has also been widely applied to IoT-based monitoring systems~\cite{Ygu-2019, jiang-2019,wang-2019}.
    In \cite{Ygu-2019}, the authors characterized the age-energy tradeoff of IoT systems by minimizing the average AoI with a given long-term average transmit power constraint and a practical truncated automatic repeat request scheme.
In \cite{jiang-2019}, the optimal policy minimizing the time-average AoI is proved to be a round-robin policy with one-packet buffers (RR-ONE).
    For the resource constrained industrial IoT networks, it was proved that the optimal stationary policy is a randomized mixture of two deterministic monotone policies \cite{wang-2019}.
     For a two-way data exchanging system with an access point and an energy-harvesting powered smart device, the timeliness limit and efficiency limit were explicitly presented in \cite{Ydong-2019} and \cite{Hu-2019}.

Moreover, there were several other information freshness measures proposed  recently \cite{Antzela-2018, Kosta-2017}.
    In \cite{Antzela-2018}, the authors characterized the cost of information staleness using the cost of update delay (CoUD).
In \cite{Kosta-2017}, the authors introduced a metric called value of information of update (VoIU) to capture the degree of importance of received information at the destination.
    More relevantly, AuD was proposed to measure the freshness of  received updates at the epochs of interest in \cite{Dong-2018}; the AuD minimizing arrival process of G/M/1/M systems was investigated in \cite{Dong-IoT-2020, Dong-2019}.
Since AuD is very applicable to many real-time IoT systems, we shall consider the timeliness of received updates of IoT systems within the AuD framework in this paper.
    Specifically,  we shall minimize the average AuD by scheduling the service process and decision process.

\section{System Model}\label{sec:model}
We consider an FCFS update-and-decide system with arrival rate $\lambda$, service rate $\mu$ and decision rate $\nu$.
    We denote the \textit{server utilization} of the system as $\rho=\frac{\lambda}{\mu}$ and assume that $\rho$ is smaller than unity so that the queue would be stable.

As shown in Fig. \ref{fig:AOI}, we denote the arrival time and the departure time of the $k$-th received update, as $t_{k}$ and $t_{k}'$ ($k$=1,2,...), respectively.
   We denote the inter-arrival time between the arrival of update $k$ and $k-1$ as $X_{k}=t_{k}-t_{k-1}$.
We also denote the system time of the $k$-th update as $T_{k}=W_{k}+S_{k}$, in which $W_{k}$ and $S_{k}$ are the waiting time and the service time of update $k$, respectively.
   During inter-departure time $Y_{k}=t_{k}'-t_{k-1}'$, we denote the number of decisions as $N_{k}$ and the epochs at which the $j$-th decisions is made as $\tau_{k_{j}}$, in which $j=1,2,...,N_{k}$.
The inter-decision time between neighboring decisions can then be expressed as $Z_{j}=\tau_{j}-\tau_{j-1}$.
    By making a decision, we mean that the received update is used to decode, to make an inference.

We consider an M/G/1/G update-and-decide system, in which the updates are generated according to a Poisson process with rate $\lambda$, the service time is generally distributed with an average of $1/\mu$, and the inter-decision time $Z_{j}$ is generally distributed with average $1/\nu$.

\begin{definition}
   \textit{(Age upon Decisions-AuD)}\cite{Dong-2018}
   At the $j$-th decision epoch $\tau_j$, the index of the most recently received update is
   \begin{align}\label{eq:tau}
   N_{\text{U}}(\tau_{j})=\text{max}\{k|t_{k}'\leq\tau_{j}\},
   \end{align}
   and the generation time of the update is
   \begin{align}\label{eq:Utau}
   U(\tau_{j})=t_{N_{\text{U}}(\tau_{j})}.
   \end{align}
   We denote the AuD of the update-and-decide system as the following random process
   \begin{align}\label{eq:delta}
   \Delta_{\text{D}}(\tau_{j})=\tau_{j}-U(\tau_{j}).
   \end{align}
\end{definition}

 Compared with AoI which evaluates information freshness at every moment $t$, i.e., $\Delta_{\text{A}}(t)=t-u(t)$, it is clear that $\Delta_{\text{D}}(\tau_{j})$ focuses only on the freshness of information at decision epoches.
   In addition, we see that AuD reduces to AoI if we replace the decision epoch $\tau_{j}$ with an arbitrary epoch $t$.

   Fig.     \ref{fig:AOI} presents an example of the arrival process, the service process, and the decision process for the evaluation of AuD.
   On one hand, in case that the inter-arrival time is relatively large and we have $X_{k}>T_{k-1}$ (e.g., $X_{3}>T_{2}$), the next update $k$ has not arrived yet by the departure time of current update $k-1$.
Therefore, the inter-departure time can be expressed as $Y_{k}=X_{k}+S_{k}-T_{k-1}$ (e.g., $Y_3=X_3+S_3-T_2$).
   On the other hand, if the inter-arrival time is relatively small and we have $X_{k}<T_{k-1}$ (e.g., $X_{2}<T_{1}$), the newly arrived update needs to wait for some time $W_k$ for its service.
In this case,  we have $T_{k}=W_k+S_{k}$ and $Y_{k}=S_{k}$ (e.g., $Y_2=S_2$).
It is also noted that there may exist several decisions during an inter-departure time.

Suppose there are $N_{T}$ decisions made during a period of $T$, the average AuD can be given by
   \begin{align}\label{eq:deltaD}
   \widebar{\Delta}_{\text{D}}=\lim_{T\rightarrow\infty}\frac{1}{ N_{\text{T}}} \sum_{j=1}^{N_{\textrm{T}}} \Delta_{\text{D}}(\tau_{j}).
   \end{align}

\begin{figure}
  \centering
  \includegraphics[width=3.0in]{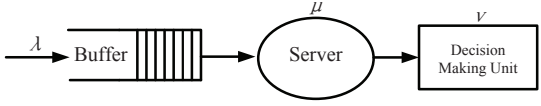}\\
  \caption{The update-and-decide system model}\label{fig:system}
\end{figure}
\begin{figure}
  \centering
  \includegraphics[width=3.3in]{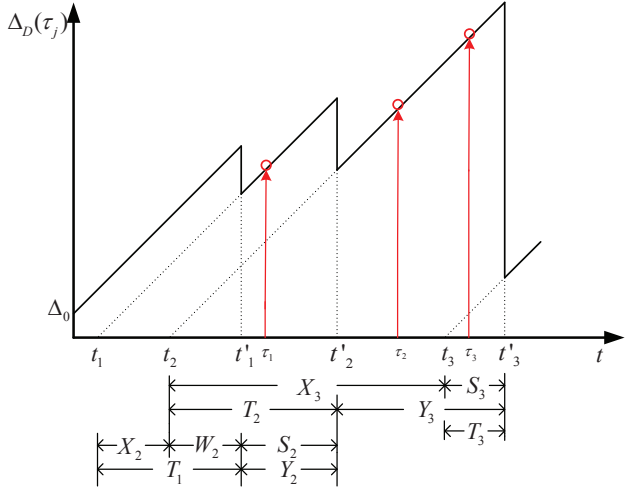}\\
  \caption{Age upon decisions}\label{fig:AOI}
\end{figure}

\section{Average AuD With Random Decisions}\label{sec:mg1}
In this section, we investigate the average AuD of an  \textit{M/G/1/M update-and-decide system} with FCFS discipline, in which the inter-arrival time is exponentially distributed with mean $1/\lambda$, the service time is generally distributed with mean $1/\mu$, and the inter-decision time is exponentially distributed with mean $1/\nu$.
    We shall evaluate the average AuD of the system under three different service time distributions, i.e., the uniform distribution, the exponential distribution, and the deterministic case.
    Also,  we shall discuss how the missing probability $p_\text{mis}$ vary with decision rate $\nu$.

\subsection{Average AuD of M/G/1/M Systems}\label{sec:mg1m}
We denote the probability density function ({pdf}) of the inter-arrival time as $f_\text{X}(x)$  and the {pdf} of the service time as $f_\text{S}(x)$.
    For an M/G/1/M system, we have  $f_\text{X}(x)=\lambda e^{-\lambda x}$.

As shown in \cite{Dong-2019}, the average AuD of a G/G/1/M update-and-decide system is given by
\begin{align}\label{eq:deltaAUD}
\widebar{\Delta}_{\text{D}}=\frac{\mathbb{E}[Y_{k}^{2}]+
2\mathbb{E}[T_{k-1}Y_{k}]}{2\mathbb{E}[Y_{k}]}.
\end{align}
Based on this result, we have the following theorem.

\begin{theorem}\label{th:MG1}
For an FCFS based M/G/1/M update-and-decide system, the average AuD is given by
\begin{align}\label{eq:deltaMG1}
\widebar{\Delta}_\text{D}^\text{M/G/1/M}=\frac{\lambda^{2}\mathbb{E}[S_{k}^{2}]+2(1-\rho)(1-\lambda\omega)}{2\lambda(1-\rho)},
\end{align}
\end{theorem}
where $\omega=\frac{\text{d}G_\text{T}(s)}{\text{d}s}|_{s=-\lambda}=\frac{(1-\rho)(G_\text{S}(-\lambda)-1)}{\lambda G_\text{S}(-\lambda)}$, $G_\text{T}(s)=\mathbb{E}[e^{sT_k}]$ and $G_\text{S}(s)=\mathbb{E}[e^{sS_k}]$ are the moment generating functions (MGF) of system time $T_{k}$ and  service time $S_k$, respectively.

\begin{proof}
To prove \textit{Theorem} \ref{th:MG1}, we need to obtain $\mathbb{E}[Y_{k}^{2}]$, $\mathbb{E}[Y_{k}]$, and $\mathbb{E}[T_{k-1}Y_{k}]$ first.
    Since $Y_{k}=X_{k}+S_{k}-T_{k-1}$ if $X_k>T_{k-1}$ and $Y_{k}=S_{k}$ if $X_k<T_{k-1}$, all of the three quantities can be obtained readily when the probability $\Pr\{X_k<T_{k-1}\}$ has been obtained.
For more details, refer to Appendix~\ref{prop:MG1}.
\end{proof}

\subsection{M/G/1/M System under Different Service Processes}\label{sec:ued}
With mean service time ${1}/{\mu}$, we denote the pdfs of the uniform distribution, the exponential distribution, and the deterministic distribution as $f_{\text{SU}}(x)$, $f_{\text{SE}}(x)$, and $f_{\text{SD}}(x)$.
In particular, we have $f_{\text{SU}}(x)=\frac{\mu}{2}$ for $x\in(0, \mu)$, $f_{\text{SE}}(x)=\mu e^{-\mu x}$ for $x>0$, and $f_{\text{SD}}(x)=\delta\left(x-\frac{1}{\mu}\right)$ for $x>0$, where $\delta(x)$ is the Dirichlet function.
Moreover, the corresponding update-and-decide system is denoted as the M/U/1/M system, the M/M/1/M system, and the M/D/1/M system, respectively.
    Based on \textit{Theorem} \ref{th:MG1}, the average AuD of these update-and-decide systems can then be calculated readily.

\begin{corollary}\label{cor:MG1M}
In M/U/1/M, M/M/1/M, M/D/1/M update-and -decide systems with arrival rate $\lambda$, service rate $\mu$, and Poisson decisions of rate $\nu$, the average AuDs are given by
\begin{align}
\widebar{\Delta}_{\text{D}}^{\text{M/U/1/M}}=&\frac{\rho(6\rho^{2}e^{2\rho}-13\rho e^{2\rho}+9e^{2\rho}+\rho-3)}{3\lambda(1-\rho)(e^{2\rho}-1)},\nonumber \\
\widebar{\Delta}_{\text{D}}^{\text{M/M/1/M}}=&\frac{\rho^{3}-\rho^{2}+1}{\lambda(1-\rho)},\nonumber \\
\widebar{\Delta}_{\text{D}}^{\text{M/D/1/M}}=&\frac{\rho^{2}+2(1-\rho)(\rho+e^{\rho}-\rho e^{\rho})}{2\lambda(1-\rho)}.\nonumber
\end{align}
\end{corollary}

\begin{proof}
By combing the results in \textit{Theorem} \ref{th:MG1}, the \textit{Corollary} \ref{cor:MG1M} can readily be proved. See Appendix~\ref{prop:MG1M}.
\end{proof}

Based on \textit{Corollary} \ref{cor:MG1M}, we then have the following theorem.
\begin{theorem} \label{MGj}
For FCFS based M/G/1/M update-and-decide systems with the common arrival rate $\lambda$ and service rate $\mu$, with uniformly distributed, exponentially distributed, and deterministic service time, respectively, we have
\begin{align}
\widebar{\Delta}_{\text{D}}^{\text{M/D/1/M}}<\widebar{\Delta}_{\text{D}}^{\text{M/U/1/M}}<\widebar{\Delta}_{\text{D}}^{\text{M/M/1/M}},
\end{align}
which is independent of the decision rate $\nu$.
\end{theorem}

\begin{proof}
See  Appendix~\ref{proof:ued}.
\end{proof}

From \textit{Theorem} \ref{MGj}, we observe that for the given common arrival rate $\lambda$ and common service rate $\mu$, the deterministic service process outperforms the uniform service process, and the uniform service process outperforms the exponential service process.
    Thus, we prefer deterministic services to random services in M/G/1/M update-and-decide systems.

\begin{figure}[htp]
  \hspace{-6 mm}
  \begin{tabular}{cc}
  \subfigure[Average AuD versus arrival rate $\lambda$, where $\mu=1.5$]
  {
  \begin{minipage}[t]{0.5\textwidth}
  \centering
  {\includegraphics[width = 3.7in] {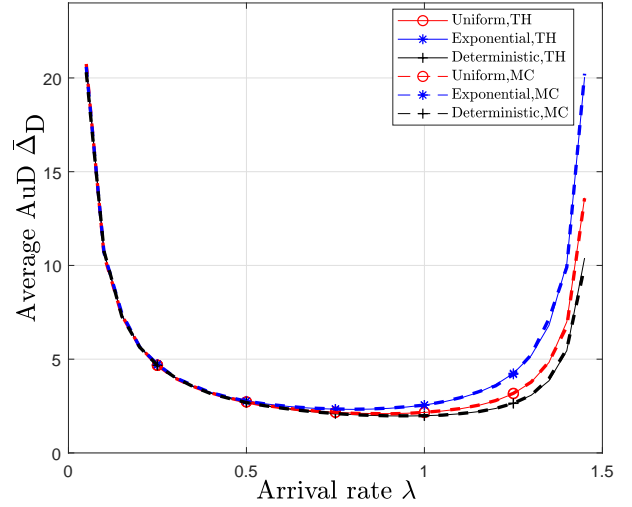} \label{fig:lambath}}
  \end{minipage}
  }\\

  \subfigure[Average AuD versus service rate $\mu$, where $\lambda=0.5$]
  {
  \begin{minipage}[t]{0.5\textwidth}
  \centering
  {\includegraphics[width = 3.7in] {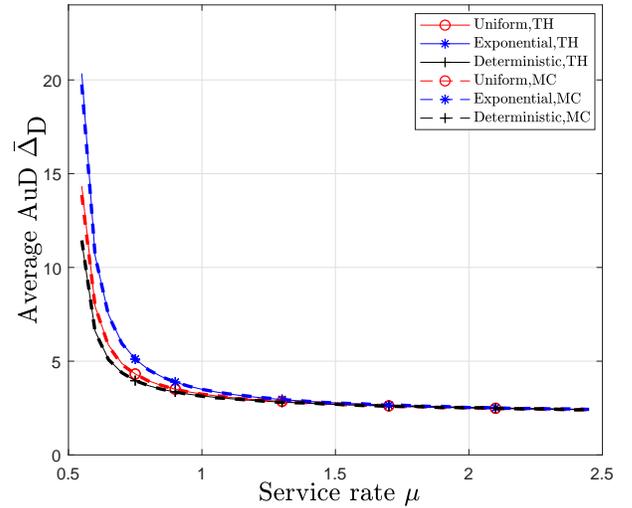} \label{fig:muth}}
  \end{minipage}
  }
  \end{tabular}
\caption{Average AuD of the systems} \label{fig:aud}
\end{figure}
\begin{figure}[!t]
  \centering
  \includegraphics[width=3.7in]{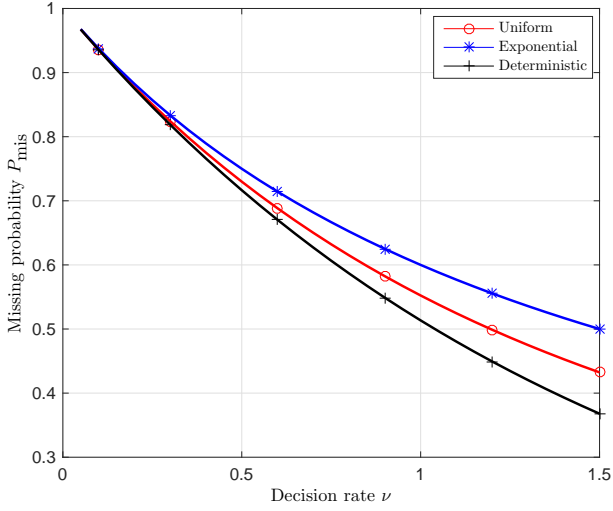}\\
  \caption{Missing probability $P_{\text{mis}}$}
  \label{fig:pmis}
\end{figure}

\subsection{Missing Probability of Updates}
As shown in \cite{Dong-2019}, although we cannot reduce the average AuD of G/G/1/M update-and-decide systems by increasing the decision rate, we do can reduce the probability of missing to use the received updates by increasing the decision rate.
    To be specific, fewer updates will be missed for making decisions when the decision rate is increased.

We define the missing probability $p_{\text{mis}}$ as the limiting ratio between the number of updates missed for decisions and number of totally received updates.
    Since the number $N_{k}$ of decisions is zero if  there is no decision made during inter-departure time $Y_{k}$, $p_{\text{mis}}$ would  be equal to probability $\Pr\{N_k=0\}$.
Thus, the missing probability $p_{\text{mis}}$ can readily be obtained by taking the expectation over $Y_{k}$, as shown in the following theorem.

\begin{theorem} \label{pmis}
In an M/G/1/M update-and-decide system, the missing probability $p_{\text{mis}}$ of updates is given by
\begin{align}\label{eq:pmis}
p_{\text{mis}}=G_{\text{S}}(-\nu) \frac{\rho \nu+\lambda}{\lambda+\nu},
\end{align}
in which $G_{\text{S}}(s)=\mathbb{E}[e^{s S_k}]$ is the MGF of service time $S_{k}$.
\end{theorem}

\begin{proof}
See  Appendix~\ref{proof:pmis}.
\end{proof}

Since the freshness and the utilization of the received updates are two of the most important concerns of IoT systems, the results on the average AuD (cf. Theorem \ref{th:MG1}, Corollary \ref{cor:MG1M}, and Theorem \ref{MGj}) and the missing probability (Theorem \ref{pmis}) provide effective evaluations for IoT based update and decide systems.

\subsection{Simulation Results}\label{sec:simulation}
In this subsection, we analyze how the arrival rate $\lambda$ and the service rate $\mu$ affect the average AuD of an M/G/1/M update-and-decide system through numerical results and Monte Carlo simulations, in which, the uniform service process, the exponential service process, and the deterministic service process are considered.

In Fig.    \ref{fig:lambath}, we present how the average AuDs of the three systems change with the common arrival rate $\lambda$, where the common service is set to $\mu=1.5$.
    First, we observe that the average AuD is large when $\lambda$ is either very small or very large.
This is because the waiting times for new updates are large when $\lambda$ is small, while the waiting times for being served are large when $\lambda$ is large.
        Second, the average AuD of the update-and-decide system is minimized when the arrival rate $\lambda$ is close to $\mu/2$.
    Third, we see that the system with a deterministic service process performs the best and the exponential process performs the worst when $\lambda$ is large.
Fourth, when $\lambda$ is small, all the three service process under test perform almost the same.
    This is because when $\lambda$ is small, the updates do not need to wait and can be delivered immediately at their arrivals, regardless of the type of the service process.

Fig.    \ref{fig:muth} presents how the average AuD changes with service rate $\mu$ when arrival rate is set to $\lambda=0.5$.
As is shown, the average AuD decreases when the service rate $\mu$ is increased.
    In particular, the average AuDs will converge to $1/\lambda$  as $\mu$ goes to infinity, i.e., is determined solely by the arrival process.
We also observe in Figs. \ref{fig:lambath} and  \ref{fig:muth} that the Monte Carlo results and the obtained theoretical results are well matched.

Fig.    \ref{fig:pmis} presents the missing probability of the system.
    We see that $p_{\text{mis}}$ would be decreased when $\nu$ is increased.
We also see that the deterministic service process performs the best and the exponential process performs the worst, especially when $\nu$ is large.

\section{Average AuD With Deterministic Decisions}\label{sec:MG1D}
In this section, we consider the performance of M/G/1/D update-and-decide systems in which the inter-decision time is deterministic.
    Unlike M/G/1/M update-and-decide systems, the decision epochs are no longer uniformly distributed within each inter-departure time $Y_{k}$.
To explicitly characterize the performance of M/G/1/D update-and-decide systems, however, we shall assume that the decision epochs are approximately uniformly distributed among inter-departure times.
    As is shown by the obtained results, this approximation only leads to negligible error.

For the given service rate $\mu$ and decision rate $\nu$, the average service time and the average inter-decision time would be $1/\mu$ and $1/\nu$, respectively.
    We assume that the decision rate $\nu$ is an integer multiple of the service rate $\mu$, i.e., $\nu=m_{0}\mu$.
To keep the missing probability $p_{\text{mis}}$ low, $m_{0}$ is set to be relatively large, i.e., $m_{0}\geq 1$.

\subsection{Average AuD of M/G/1/D Update-and-Decide Systems}\label{sec:mg1d}
We denote the {pdf} of the inter-arrival time as $f_\text{X}(x)$, the {pdf} of the service time as $f_\text{S}(x)$, and the {pdf} of the service time as $f_\text{Z}(x)$.
    Thus, we have $f_\text{X}(x)=\lambda e^{-\lambda x}$, and $f_{\text{Z}}(x)=\delta\left(x-\frac{1}{\nu}\right)$.

\begin{theorem} \label{th:MG1D}
In an FCFS M/G/1/D update-and-decide system, the average AuD is given by
\begin{align}\nonumber
&\widebar{\Delta}_\text{D}^\text{M/G/1/D}\\
&~~~=\frac{\lambda\rho}{\nu}\bigg(\mathbb{E}[T_{k-1}|X_{k}\leq T_{k-1}]\mathbb{E}[N_{k}^{1}]+\frac{\mathbb{E}[(N_{k}^{1})^{2}]}{2\nu}\bigg)\nonumber \\
&~~~~~+\frac{\lambda(1-\rho)}{\nu}\bigg(\mathbb{E}[T_{k-1}|X_{k}>T_{k-1}](\mathbb{E}[N_{k}^{2}]+\mathbb{E}[N_{k}^{3}])\nonumber \\
\label{eq:deltaMG1D}
&~~~~~+\frac{\mathbb{E}[(N_{k}^{2})^2]+\mathbb{E}[(N_{k}^{3})^2]+2(\mathbb{E}[N_{k}^{2}]+\mathbb{E}[N_{k}^{3}])}{2\nu}\bigg),
\end{align}
\end{theorem}
in which $N_{k}^{1}$ is the number of decisions made during $Y_{k}$ conditioned on $X_{k}<T_{k-1}$, $N_{k}^{2}$ and $N_{k}^{3}$ are the numbers of decisions made during the parts of $Y_{k}$ before and after the arrival of update $k$, in the case of $X_{k}>T_{k-1}$.

\begin{proof}
See  Appendix~\ref{proof:MG1D}.
\end{proof}

In particular, the terms $\mathbb{E}[T_{k-1}|X_{k}\leq T_{k-1}]$ and $\mathbb{E}[T_{k-1}|X_{k}>T_{k-1}]$ can be given by the following proposition.

\begin{proposition}\label{pro:tk}
The expectations of $T_{k}$ conditioned on $X_{k}<T_{k-1}$ and $X_{k}>T_{k-1}$ are respectively given by
\begin{align}\label{tk1}
\mathbb{E}[T_{k-1}|X_{k}\leq T_{k-1}]=&\frac{\mathbb{E}[T_{k}]+w}{\rho}, \nonumber \\
\mathbb{E}[T_{k-1}|X_{k}>T_{k-1}]=&-\frac{\omega}{1-\rho},
\end{align}
where $\omega=\frac{\text{d}G_\text{T}(s)}{\text{d}s}|_{s=-\lambda}=\frac{(1-\rho)(G_\text{S}(-\lambda)-1)}{\lambda G_\text{S}(-\lambda)}$, $G_\text{T}(s)$ and $G_\text{S}(s)$ are the MGFs of the system time $T_{k}$ and the service time $S_k$.
\end{proposition}

\begin{proof}
See  Appendix~\ref{proof:tk}.
\end{proof}

Similar to the analysis for the M/G/1/M update-and-decide systems, we shall investigate the M/G/1/D update-and-decide systems under three typical distributions of service time $S_{k}$ as shown in the following corollary.


\begin{corollary}\label{cor:MG1D}
In M/U/1/D, M/M/1/D, M/D/1/D update-and-decide systems with arrival rate $\lambda$, service rate $\mu$, and periodic decisions of rate $\nu$, the average AuDs are given by
\begin{align}
\widebar{\Delta}_\text{D}^\text{M/U/1/D}=&\frac{2e^{2\rho}(1-\rho)}{\mu(e^{2\rho}-1)}+\frac{(1-\rho)(1+u_{0})}{2\mu m_{0}(1-u_{0})}-\frac{1}{\mu\rho(1-\rho)}\nonumber \\
&-\frac{(2m_{0}^{2}+1)\rho^{2}+(16m_{0}^2-1)\rho-36m_{0}^{2}-6m_{0}}{12m_{0}^{2}\mu(1-\rho)},\nonumber \\
\widebar{\Delta}_\text{D}^\text{M/M/1/D}=&\frac{2\rho^2-3\rho+2}{\mu(1-\rho)}+\frac{\rho(1+\omega_{0})}{2\mu m_{0}(1-\omega_{0})}+\frac{(1-\rho)(1+u_{0})}{2\mu m_{0}(1-u_{0})},\nonumber \\
\widebar{\Delta}_\text{D}^\text{M/D/1/D}=&\frac{e^{\rho}(1-\rho)}{\mu\rho}+\frac{(1-\rho)(1+u_{0})}{2\mu m_{0}(1-u_{0})}+\frac{-3\rho^2+6\rho-2}{2\mu\rho(1-\rho)},\nonumber
\end{align}
\end{corollary}
where, $\omega_{0}=e^{-\frac{\mu}{\nu}}$ and $u_{0}=e^{-\frac{\lambda}{\nu}}$.
\begin{proof}
See  Appendix~\ref{proof:Average}.
\end{proof}

Based on \textit{Corollary} \ref{cor:MG1D}, we  have the following proposition.
\begin{proposition} \label{MG1Dj}
With the same arrival rate $\lambda$, service rate $\mu$, and decision rate $\nu=m_{0}\mu$, we have
\begin{align}
&\widebar{\Delta}_{\text{D}}^{\text{M/U/1/D}}\geq \widebar{\Delta}_{\text{D}}^{\text{M/U/1/M}}, \nonumber \\
&\widebar{\Delta}_{\text{D}}^{\text{M/M/1/D}}\geq\widebar{\Delta}_{\text{D}}^{\text{M/M/1/M}}, \nonumber \\
&\widebar{\Delta}_{\text{D}}^{\text{M/D/1/D}}\geq\widebar{\Delta}_{\text{D}}^{\text{M/D/1/M}},
\end{align}
with equalities hold as $m_{0}$ goes to infinity.
\end{proposition}

\begin{proof}
It can be readily verified that $(1+\omega_{0})/(m_{0}(1-\omega_{0}))$ and $(1+u_{0})/(m_{0}(1-u_{0}))$ are equal to $(1+2/(1/\omega_{0}-1))/m_{0}$ and $(1+2/(1/u_{0}-1))/m_{0}$, respectively.
    Since the items $2/(1/\omega_{0}-1)$ and $2/(1/u_{0}-1)$ are increasing with $m_{0}$ more slowly than linear, the average AuDs shown in \textit{Corollary} \ref{cor:MG1D} would be decreasing with $m_{0}$.

Also note that the denominators $m_{0}(1-\omega_{0})$ and $m_{0}(1-u_{0})$ can be rewritten as $(1-\text{exp}(-1/m_{0}))/(1/m_{0})$ and $(1-\text{exp}(-\rho/m_{0}))/(1/m_{0})$.
As $m_{0}$ goes to infinity, therefore, we have $m_{0}(1-\omega_{0})=1$ and $m_{0}(1-u_{0})=\rho$.
    By combing these limits and the results in \textit{Corollary} \ref{cor:MG1D}, the proposition can be proved readily.
\end{proof}

Therefore, it is concluded that for the same service process, the M/G/1/M update-and-decide system perform better than the corresponding M/G/1/D update-and-decide system.

\begin{proposition} \label{prop:MG1Dbijiao}
With the same arrival rate $\lambda$, service rate $\mu$, and decision rate $\nu=m_{0}\mu$, as, we have
\begin{align}
\left\{
        \begin{aligned}
            &\widebar{\Delta}_{\text{D}}^{\text{M/D/1/D}}
                    <\widebar{\Delta}_{\text{D}}^{\text{M/U/1/D}}
                    <\widebar{\Delta}_{\text{D}}^{\text{M/M/1/D}}  &&\text{if}~ m_{0}> m_{0}^{*},\\
            &\widebar{\Delta}_{\text{D}}^{\text{M/D/1/D}}
                    <\widebar{\Delta}_{\text{D}}^{\text{M/M/1/D}}
                    <\widebar{\Delta}_{\text{D}}^{\text{M/U/1/D}}  &&\text{if}~ m_{0}\leq m_{0}^{*},
        \end{aligned}
\right.
\end{align}
in which $m_0^*$ is a positive integer.
\end{proposition}
\begin{proof}
See  Appendix~\ref{proof:MG1Dbijiao}.
\end{proof}

We also have the following results on the missing probability of M/G/1/D update-and-decide systems.
\begin{proposition} \label{prop:MG1Dpmis}
In M/U/1/D, M/M/1/D, and M/D/1/D update -and-decide systems with the common arrival rate $\lambda$, service rate $\mu$, and periodic decisions with rate $\nu=m_{0}\mu$, the missing probability of updates are, respectively, given by
\begin{align}\label{MG1DPmis}
p_{\text{mis}}^{\text{M/U/1/D}}&=\frac{1}{8m_{0}}+\frac{(1-\rho)(m_{0}(1-u_{0})-\rho)}{4\rho^{2}},\nonumber \\
p_{\text{mis}}^{\text{M/M/1/D}}&=\frac{1}{2}+\frac{m_{0}(u_{0}-1)}{2\rho},\nonumber \\
p_{\text{mis}}^{\text{M/D/1/D}}&=0. \nonumber
\end{align}
\end{proposition}
\begin{proof}
See  Appendix~\ref{proof:MG1DPmis}.
\end{proof}

From proposition \ref{prop:MG1Dpmis}, it is clear that the missing probability is decreasing with $m_{0}$ in M/U/1/D and M/M/1/D systems, and the M/D/1/D system performs the best.

\begin{figure}[htp]
  \hspace{-6 mm}
  \begin{tabular}{cc}
  \subfigure[Average AuD versus arrival rate $\lambda$, where $\mu=1.5$, $m_{0}=20$]
  {
  \begin{minipage}[t]{0.5\textwidth}
  \centering
  {\includegraphics[width = 3.5in] {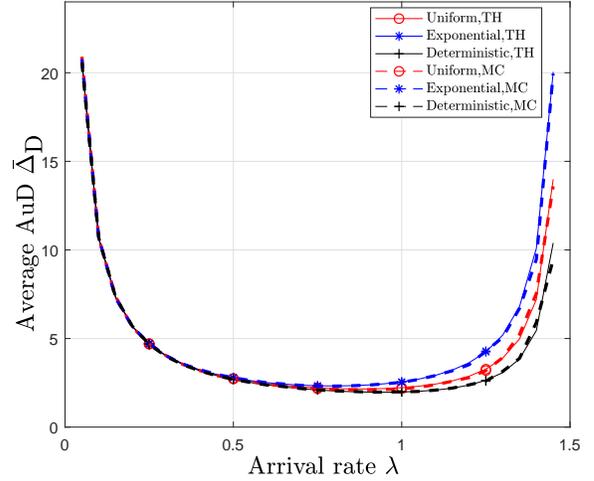} \label{fig:MG1D}}
  \end{minipage}
  }\\

  \subfigure[Average AuD versus service rate $\mu$, where $\lambda=0.5$, $m_{0}=20$]
  {
  \begin{minipage}[t]{0.5\textwidth}
  \centering
  {\includegraphics[width = 3.5in] {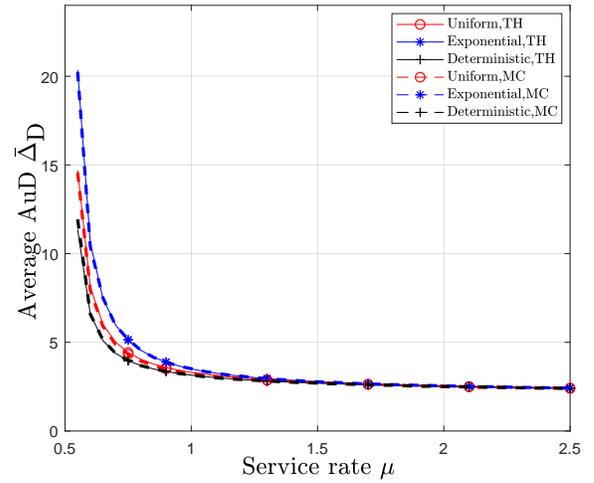} \label{fig:MG1DMU}}
  \end{minipage}
  }
  \end{tabular}
\caption{Average AuD of the systems} \label{fig:aud}
\end{figure}

\begin{figure}[!t]
  \centering
  \includegraphics[width=3.5in]{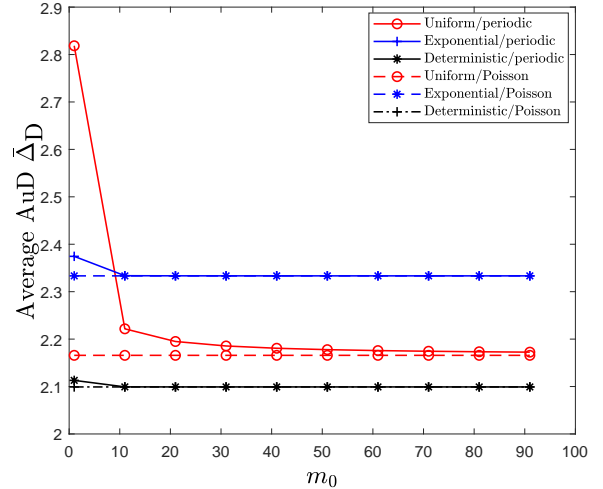}\\
  \caption{Average AuD versus decision rate}
  \label{fig:mMG1D}
\end{figure}

\begin{figure}[!t]
  \centering
  \includegraphics[width=3.5in]{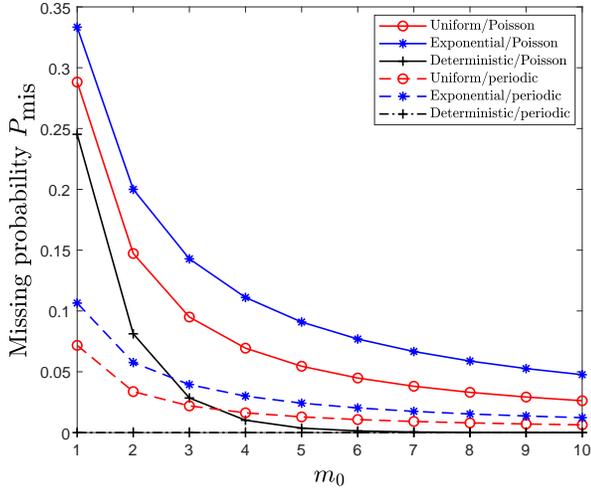}\\
  \caption{Missing probability $P_{mis}$ versus decision rate}
  \label{fig:MG1DPmis}
\end{figure}

\subsection{Simulation Results}\label{sec:MG1Dsimulation}
In this subsection, we investigate the performance of update-and-decide systems with periodic decisions through numerical results and Monte Carlo simulations.

Fig.  \ref{fig:MG1D}    plots how the average AuDs of M/U/1/D, M/M/1/D and M/D/1/D systems vary with the common arrival rate $\lambda$, in which the service rate is set to $\mu=1.5$ and the decision rate is set to $\nu=30$, i.e., $m_0=20$.
    We observe that the average AuD is large when $\lambda$ is either small or relatively large, since the waiting time for a new update or crowded queue is large, respectively.
Fig. \ref{fig:MG1DMU} presents how the average AuD changes with service rate $\mu$, in which the arrival rate is set to $\lambda=0.5$ and the decision rate is also set to $\nu=30$.
    We observe that for each system, the average AuD decreases as the service rate $\mu$ is increased.
Due to the approximation used in the calculation of the average AuD (cf. \eqref{eq:deltaMG1D}), we observe that the theoretical results for M//U/1/D systems (the dashed curve with $\circ$) slightly deviate from the corresponding Monte Carlo results in Fig.  \ref{fig:MG1D}.
    Also, it is observed that the deviation vanishes as $m_{0}$ goes to infinity.

In Figs.  \ref{fig:mMG1D} and \ref{fig:MG1DPmis}, we investigate the performance comparisons between systems using Poisson decisions and systems with periodic decisions.
    Particularly, we set the arrival rate to $\lambda=0.75$ and the service rate to $\mu=1.5$.
In Fig.  \ref{fig:mMG1D}, we observe that the average AuD of systems with Poisson decisions are larger than that of corresponding systems with periodic decisions when $m_0$ is relatively small.
    As $m_{0}$ is increased to be relatively large, however, Poisson decisions and periodic decisions  yield the same average AuDs.
Also, it is seen that with periodic decisions, the system using periodic arrivals has the smallest average AuD, and the system with uniform arrivals performs better than the system with exponential arrivals in most cases, except the cases when $m_0$ is relatively small.
    Fig.  \ref{fig:MG1DPmis} presents the missing probabilities of the M/U/1/D and the M/M/1/D systems.
As is shown, the missing probabilities are also decreasing with $m_{0}$.
    It is also observed that for a system with Poisson arrivals, the periodic decision process outperforms other random decision processes in terms of missing probability.

\section{Conclusion and Decisions}\label{sec:conclusion}
In this paper, we investigated the timeliness and the utilization of the received updates in an IoT-based update-and-decide system, in which a Poisson arrival process and a general service process are used.
    In particular, we are interested in whether the random decision process outperforms the deterministic decision process and what kind of service process minimizes the average AuD of update-and-decide systems.
We showed that when the arrival rate $\lambda$ is small or relatively large, the average AuD will be large;
    when service rate $\mu$ is increased, the average AuD will be reduced substantially.
Also, we showed that the deterministic service process outperforms the uniform service process, which further outperforms the exponential service process, for both systems using the Poisson decision process and the periodic decision process.
\begin{figure}
  \centering
  \includegraphics[width=3.9in]{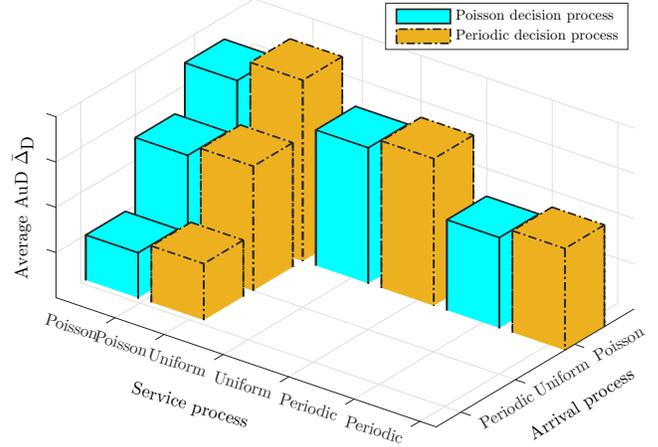}\\
  \caption{Comparison of average AuDs.}\label{fig:zhuxingtu}
\end{figure}


By combing the results of this paper and the results obtained in \cite{Dong-IoT-2020}, we present the comparisons of the average AuDs of several update-and-decide systems in Fig. \ref{fig:zhuxingtu} and Table \ref{tb:aud_comprison}.
    In terms of average AuD, it is concluded that
    \begin{itemize}
      \item the periodic (deterministic) decision process underperforms the Poisson (random) process, as shown in Proposition \ref{MG1Dj}.
            As the decision rate is sufficiently large, periodic and Poisson decision processes perform equally well, as shown in Proposition \ref{MG1Dj}.
         Intuitively,  the randomness in the decision process can provide some flexibility for monitor to timely using some of the recently received updates;
      \item the periodic service process outperforms random service processes (e.g., the uniform/Poisson process), as shown in Theorem \ref{MGj} and Proposition \ref{prop:MG1Dbijiao}.
            Also, the periodic arrival process outperforms random arrival processes \cite{Dong-IoT-2020}.
          Note that if the arrivals/services are less random, the uncertainty in the update reception process at monitor would also be reduced, which is beneficial for the timeliness of decision makings;
      \item a Poisson process limits the timeliness of the system more if it is the arrival process than it is the service process.
            For example, a system using the Poisson services performs much better than a system using Poisson arrivals in most cases, as shown in the 4th/5th rows and 2th/3th rows in Table \ref{tb:aud_comprison}, respectively.
    \end{itemize}

\vspace{3mm}
\begin{table}[htbp] 
\centering
    \caption{The average AuDs of IoT-based update-and-decide systems, $\lambda=0.75, \mu=1.5, \nu=15$.}\label{tb:aud_comprison}
{\small
\begin{tabular}{|c|c|c|c|c|c|}
\bottomrule
\multicolumn{2}{|c|}{\diagbox{}{Sv./Arv.$^{\rm a}$}}& Poisson &Uniform & Periodic & \\
\hline
\multirow{2}{*}{\tabincell{c}{Poisson\\Arrivals} }
        &\tabincell{c}{Poisson\\Decisions}&2.3333 & 2.1658& 2.0091 &smaller\\
        \cline{2-5}  
        &\tabincell{c}{Perodic\\Decisions}&2.3337& 2.2640& 2.0993 & larger\\
\hline
\multirow{2}{*}{\tabincell{c}{Poisson\\Services}}
        &\tabincell{c}{Poisson\\Decisions} &2.3333 & 1.7870  & 1.5028 & smaller\\
        \cline{2-5}  
        &\tabincell{c}{Perodic\\Decisions}&2.3336& 1.7892& 1.5037& larger \\
\hline
       \multicolumn{2}{|c|}{}&largest &middle &smallest&\\
\bottomrule
\end{tabular}
}
\begin{tablenotes}
        \footnotesize
        \item{$^{\rm a}$}Specifying the service process for the second and the third row; specifying the arrival process for the fourth and the fifth row.
      \end{tablenotes}
\end{table}

\appendix
\subsection{Proof of \textit{Theorem} \ref{th:MG1}}\label{prop:MG1}
\begin{proof} \label{proof:MG1}
From \cite[Chap. 14, pp. 522]{Stewart-2019}, the MGF of the system time $T_k$ can be given by
\begin{align}
G_\text{T}(s)=\frac{-s(1-\rho)G_\text{S}(s)}{-s-\lambda+\lambda G_\text{S}(s)}.
\end{align}

Since the inter-departure time $Y_{k}$ can be rewritten by
\begin{align}  \label{eq:yk}
         Y_k=\left\{
                   \begin{aligned}
                   &S_k,        &\text{if}~ X_k \leq T_{k-1} , \\
                   &X_k + S_k - T_{k-1} ,  &\text{if}~ X_k > T_{k-1},
                   \end{aligned}
   \right.
\end{align}
the probability of $X_{k}\leq T_{k-1}$ can then be given by
\begin{align} \label{eq:tkxk}
 \Pr\{X_{k}\leq T_{k-1}\}&=\int_{0}^{\infty}f_{\text{T}}(t)dt\int_{0}^{t}f_{\text{X}}(x)dx \nonumber \nonumber \\
                      &=\int_{0}^{\infty}f_{\text{T}}(t)(1-e^{-\lambda t})dt \nonumber \\
                      &=\rho.
\end{align}

By taking the expectation over $Y_{k}$, we have
\begin{align}
\mathbb{E}[Y_{k}]&=\rho \mathbb{E}[Y_{k}|X_{k}\leq T_{k-1}]+(1-\rho)\mathbb{E}[Y_{k}|X_{k}>T_{k-1}]\nonumber \\
        &=(1-\rho)\mathbb{E}[X_{k}-T_{k-1}|X_{k}>T_{k-1}]+\mathbb{E}[S_{k}].
\end{align}

Moreover, the expectation of $X_{k}-T_{k-1}$ conditioned on $X_{k}>T_{k-1}$ is given by
\begin{align}
&\mathbb{E}[X_{k}-T_{k-1}|X_{k}>T_{k-1}]\nonumber \\
&=\frac{1}{1-\rho}\int_{0}^{\infty}f_{\text{T}}(t)dt\int_{t}^{\infty}(x-t)f_{\text{X}}(x)dx\nonumber \\
                                  &=\frac{1}{1-\rho}\int_{0}^{\infty}\frac{1}{\lambda}f_{\text{T}}(t)e^{-\lambda t}dt\nonumber \\
                                  &=\frac{1}{\lambda(1-\rho)}G_{\text{T}}(s)|_{s=-\lambda}\nonumber \\
                                  &=\frac{1}{\lambda}.
\end{align}

Hence, we have
\begin{align}\label{eq:Eyk}
\mathbb{E}[Y_{k}]&=(1-\rho)\mathbb{E}[X_{k}-T_{k-1}\mid X_{k}>T_{k-1}]+\mathbb{E}[S_{k}]\nonumber \\
        &=\frac{1-\rho}{\lambda}+\frac{1}{\mu}\nonumber \\
        &=\frac{1}{\lambda}.
\end{align}

Likewise, we have
\begin{align}
&\mathbb{E}[(X_{k}-T_{k-1})^{2}|X_{k}>T_{k-1}]\nonumber \\ &=\frac{1}{1-\rho}\int_{0}^{\infty}f_{\text{T}}(t)dt\int_{t}^{\infty}(x-t)^{2}f_{\text{X}}(x)dx\nonumber \\
&=\frac{1}{1-\rho}\int_{0}^{\infty}\frac{2}{\lambda^{2}}f_{\text{T}}(t)e^{-\lambda t}dt\nonumber \\
&=\frac{2}{\lambda^{2}}
\end{align}
and
\begin{align}\label{eq:Eyk2}
\mathbb{E}[Y_{k}^{2}]=&\rho \mathbb{E}[Y_{k}^{2}|X_{k}\leq T_{k-1}]+(1-\rho)\mathbb{E}[Y_{k}^{2}|X_{k}>T_{k-1}]\nonumber \\
=&\mathbb{E}[S_{k}^{2}]+(1-\rho)\mathbb{E}[(X_{k}-T_{k-1})^{2}|X_{k}>T_{k-1}]\nonumber \\
            &+2(1-\rho)\mathbb{E}{S_{k}}\mathbb{E}[X_{k}-T_{k-1}|X_{k}>T_{k-1}]\nonumber \\
=&\mathbb{E}[S_{k}^{2}]+(1-\rho)(\frac{2}{\lambda^{2}}+\frac{2}{\lambda\mu})\nonumber \\
=&\mathbb{E}[S_{k}^{2}]+\frac{2}{\lambda^{2}}-\frac{2}{\mu^{2}}.
\end{align}

We also have
\begin{align}
&\mathbb{E}[T_{k-1}X_{k}-T_{k-1}^{2}|X_{k}>T_{k-1}]\nonumber \\
&=\frac{1}{1-\rho}\int_{0}^{\infty}f_{\text{T}}(t)dt\int_{t}^{\infty}(xt-t^2)f_{\text{X}}(x)dx\nonumber \\
                                         &=-\frac{1}{\lambda(1-\rho)}\frac{dG_{\text{T}}(s)}{ds}|_{s=-\lambda}
\end{align}
and
\begin{align}\label{eq:Etkyk}
\mathbb{E}[T_{k-1}Y_{k}]=&\Pr\{X_{k}\leq T_{k-1}\} \mathbb{E}[T_{k-1}Y_{k}|X_{k}\leq T_{k-1}]\nonumber \\
&+\Pr\{X_{k}>T_{k-1}\}\mathbb{E}[T_{k-1}Y_{k}|X_{k}>T_{k-1}]\nonumber \\
=&\mathbb{E}[T_{k-1}]\mathbb{E}[S_{k}]-\frac{1}{\lambda}\frac{dG_{\text{T}}(s)}{ds}|_{s=-\lambda}\nonumber \\
=&\frac{1}{\mu^{2}}+\frac{\lambda \mathbb{E}[S_{k}^{2}]}{2(\mu-\lambda)}-\frac{1}{\lambda}\frac{dG_{\text{T}}(s)}{ds}|_{s=-\lambda}.
\end{align}

By combining \eqref{eq:deltaAUD}, \eqref{eq:Eyk}, \eqref{eq:Eyk2} and \eqref{eq:Etkyk}, the proof of \textit{Theorem} \ref{th:MG1} would be completed.
 \end{proof}

\subsection{Proof of \textit{Corollary} \ref{cor:MG1M}}\label{prop:MG1M}
\begin{proof}
For the uniformly distributed service time with parameter $\frac{2}{\mu}$, the pdf $f_{\text{S}}(x)$ is given by
\begin{align}\label{eq:SU}
f_{\text{SU}}(x)=\frac{\mu}{2}.
\end{align}

Moreover, we have
$\mathbb{E}[S_{\text{U}k}]=\frac{1}{\mu}$, $\mathbb{E}[S_{\text{U}k}^{2}]=\frac{4}{3\mu^{2}}$, $w=\frac{1-\rho}{\lambda}-\frac{2(1-\rho)(e^{2\rho})}{\mu(e^{2\rho}-1)}$, and $G_{\text{SU}}(-\lambda)=\frac{1-e^{-2\rho}}{2\rho}$.

Based on \eqref{eq:deltaAUD}, the average AuD of a system with uniformly distributed service time can thus be obtained as
\begin{align}\label{eq:DU}
\widebar{\Delta}_{\text{D}}^{\text{M/U/1/M}}=\frac{\rho(6\rho^{2}e^{2\rho}-13\rho e^{2\rho}+9e^{2\rho}+\rho-3)}{3\lambda(1-\rho)(e^{2\rho}-1)}.
\end{align}

When the service time is exponentially distributed with parameter $\mu$ $f_{\text{S}}(x)$, we have
\begin{align}\label{eq:SE}
f_{\text{SE}}(x)=\mu e^{-\mu x}.
\end{align}

In this case, we have
$\mathbb{E}[S_{\text{E}k}]=\frac{1}{\mu}$, $\mathbb{E}[S_{\text{E}k}^{2}]=\frac{2}{\mu^{2}}$, $w=\frac{\rho-1}{\mu}$, and $G_{\text{SE}}(-\lambda)=\frac{1}{\rho+1}$.

By using \eqref{eq:deltaAUD}, the corresponding average AuD  is given by
\begin{align}\label{eq:DE}
\widebar{\Delta}_{\text{D}}^{\text{M/M/1/M}}=\frac{\rho^{3}-\rho^{2}+1}{\lambda(1-\rho)}.
\end{align}

When the service time is equal to $\frac{1}{\mu}$ deterministically, the pdf $f_{\text{S}}(x)$ can be presented as
\begin{align}\label{eq:SD}
f_{\text{SD}}(x)=\delta\left(x-\frac{1}{\mu}\right),
\end{align}
in which $\delta(x)$ is the Dirichlet function.

For this case, we have
$\mathbb{E}[S_{\text{D}k}]=\frac{1}{\mu}$, $\mathbb{E}[S_{\text{D}k}^{2}]=\frac{1}{\mu^{2}}$, $w=\frac{(1-\rho)(1-e^{\rho})}{\lambda}$, $G_{\text{SD}}(-\lambda)=e^{-\rho}$, and
\begin{align}\label{eq:DD}
\widebar{\Delta}_{\text{D}}^{\text{M/D/1/M}}=\frac{\rho^{2}+2(1-\rho)(\rho+e^{\rho}-\rho e^{\rho})}{2\lambda(1-\rho)}.
\end{align}

Thus, the corollary is proved.
\end{proof}

\subsection{Proof of \textit{Theorem} \ref{MGj}}\label{proof:ued}
\begin{proof}
From Appendix~\ref{prop:MG1M}, the second moment of the service time with the uniform, exponential, and deterministic distribution are, respectively, given by $\mathbb{E}[S_{\text{U}k}^{2}]=\frac{4}{3\mu^{2}}$, $\mathbb{E}[S_{\text{E}k}^{2}]=\frac{2}{\mu^{2}}$, and $\mathbb{E}[S_{\text{D}k}^{2}]=\frac{1}{\mu^{2}}$.
    Thus, it can be seen that $\mathbb{E}[S_{\text{D}k}^{2}]>\mathbb{E}[S_{\text{U}k}^{2}]>\mathbb{E}[S_{\text{D}k}^{2}]$.

By defining the following auxiliary function,
\begin{align}
f_{1}(\rho)&=G_{\text{SE}}(-\lambda)-G_{\text{SU}}(-\lambda)\\
&=\frac{1}{\rho+1}-\frac{1-e^{-2\rho}}{2\rho}\nonumber \\
&=\frac{2\rho-(\rho+1)(1-e^{-2\rho})}{2\rho(\rho+1)},
\end{align}
we observe that $2\rho(\rho+1)>0$ if $0<\rho<1$.

    We further denote $g(\rho)=2\rho-(\rho+1)(1-e^{-2\rho})$, for which the derivatives over $\rho$ is $g'(\rho)=1-(2\rho+1)e^{-2\rho}$ and $g''(\rho)=4\rho e^{-2\rho}$.
Since $0<\rho<1$, we have $g''(\rho)>0$, which means that $g'(\rho)$ is monotonically increasing with $\rho$.
    In particular, we have  $0<g'(\rho)<1-3e^{-2}$, which means that $g(\rho)$  is monotonically increasing with $\rho$ and $0<g(\rho)<2e^{-2}$.
Therefore, we have $f_{1}(\rho)>0$, i.e., $G_{\text{SE}}(-\lambda)>G_{\text{SU}}(-\lambda)$.

Likewise, we defined the second auxiliary function as
\begin{align}
f_{2}(\rho)&=G_{\text{SU}}(-\lambda)-G_{\text{SD}}(-\lambda)\nonumber \\
&=\frac{1-e^{-2\rho}}{2\rho}-e^{-\rho}\nonumber \\
&=\frac{1-e^{-2\rho}-2\rho e^{-\rho}}{2\rho}
\end{align}
and denote $h(\rho)=1-e^{-2\rho}-2\rho e^{-\rho}$.
    We then  have $h'(\rho)=2e^{-\rho}(e^{-\rho}-1+\rho)$.
    Denote $l(\rho)=e^{-\rho}-1+\rho$, then we have $l'(x)=1-e^{-\rho}$.
It is clear that $0<l'(x)<1-e^{-1}$, which means that $h'(\rho)>0$.
    Therefore,  $h(\rho)$ is monotonically increasing with $\rho$ and satisfies $0<h(\rho)<\frac{e^{2}-1-2e}{e^{2}}$.
Hence, we have $f_{2}(\rho)>0$, i.e., $G_{\text{SU}}(-\lambda)>G_{\text{SD}}(-\lambda)$.

From equation \eqref{eq:deltaMG1}, it is observed that the average AuD would be decreased either if $E[S^{2}]$ is decreased or $w$ is increased.
    Therefore, we have $\widebar{\Delta}_{\text{D}}^{\text{M/D/1/M}}<\widebar{\Delta}_{\text{D}}^{\text{M/U/1/M}}<\widebar{\Delta}_{\text{D}}^{\text{M/M/1/M}}$.

This completes the proof of the \textit{Theorem} \ref{MGj}.
\end{proof}

  \begin{figure}[htp]
  \hspace{-6 mm}
  \begin{tabular}{cc}
  \subfigure[]
  {
  \begin{minipage}[t]{0.5\textwidth}
  \centering
  {\includegraphics[width = 2.0in] {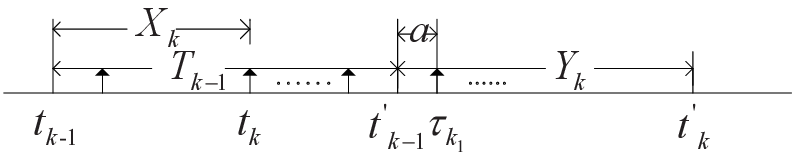} \label{fig:xiaoyu}}
  \end{minipage}
  }\\

  \subfigure[]
  {
  \begin{minipage}[t]{0.5\textwidth}
  \centering
  {\includegraphics[width = 2.0in] {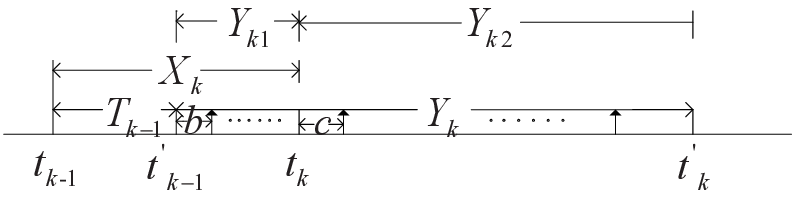} \label{fig:dayu}}
  \end{minipage}
  }
  \end{tabular}
\caption{Inter-arrival time and system time for systems} \label{fig:xktk}
\end{figure}

 \subsection{Proof of \textit{Theorem} \ref{pmis}}\label{proof:pmis}
 \begin{proof} \label{proof:pmis}
 The missing probability can be derived as follows,
 \begin{align}
 p_{\text{mis}}&=\mathbb{E}[e^{-\nu Y_{k}}]\nonumber \\
        &=\rho \mathbb{E}[e^{-\nu S_{k}}]+(1-\rho)\mathbb{E}[e^{-\nu(X_{k}-T_{k-1}+S_{k})}|X_{k}>T_{k-1}]\nonumber \\
        &=\mathbb{E}[e^{-\nu S_{k}}](\rho+\int_{0}^{\infty}f_{\text{T}}(t)dt\int_{t}^{\infty}e^{-\nu(x-t)}f_{\text{X}}(x)dx)\nonumber \\
        &=G_{\text{S}}(-\nu)\Big(\frac{\rho \nu+\lambda}{\lambda+\nu}\Big).
 \end{align}

 This completes the proof of the \textit{Theorem} \ref{pmis}.
\end{proof}

\subsection{Proof of \textit{Theorem} \ref{th:MG1D}}\label{proof:MG1D}
\begin{proof} \label{proof:MG1D}
Firstly, we consider the case of $X_{k}\leq T_{k-1}$, as shown in Fig. \ref{fig:xiaoyu}.
    Suppose $T_{k-1}$ consists of $j$ decision intervals, i.e., $\frac{j}{\nu}\leq T_{k-1}<\frac{j+1}{\nu}$.
We denote $\tau_{j}=T_{k-1}-\frac{j}{\nu}$ and $a_{j}=\frac{1}{\nu}-\tau_{j}$.
    In particular, we assume that $a$ is (approximately) uniformly distributed with parameter $1/\nu$, i.e.,
\begin{align}
f_{\text{a}}(x)=\nu, x\in\Big(0,\frac{1}{\nu}\Big).\nonumber
\end{align}

We denote the number of decisions made during $Y_{k}$ on condition $X_{k}\leq T_{k-1}$ as $N_{k}^{1}$.
    Since the AuD of $i$-th made can be written as
\begin{align}
\Delta_\text{D}(\tau_{k_{i}})=T_{k-1}+a+\frac{i-1}{\nu},   i=1, 2, \cdot\cdot\cdot , N_{k}^{1}, \nonumber
\end{align}
 the expected sum AuD during $Y_{k}$ would be
\begin{align}\label{eq:xiaoyu}
&\mathbb{E}[\Delta_\text{Dk}^{1}|X_{k}\leq T_{k-1}]\nonumber \\
&=\mathbb{E}[\Sigma_{i=1}^{N_{k}^{1}}\Delta_\text{D}(\tau_{k_{i}})|X_{k}\leq T_{k-1}]\nonumber \\
&=\mathbb{E}[T_{k-1}|X_{k}\leq T_{k-1}]\mathbb{E}[N_{k}^{1}]+\frac{\mathbb{E}[(N_{k}^{1})^{2}]}{2\nu}.
\end{align}

Secondly, we consider the case of $X_{k}>T_{k-1}$, as shown in Fig.  \ref{fig:dayu}.
    We denote the parts before and after the arrival of updates $k$ as $Y_{k1}$ and $Y_{k2}$, i.e., $Y_{k1}=X_{k}-T_{k-1}$ and $Y_{k2}=S_{k}$.
We denote the number of decisions made during $Y_{k1}$ and $Y_{k2}$ as $N_{k}^{2}$ and $N_{k}^{3}$.
    We approximately denote the length between arrival epoch $t_{k}$ and the next decision epoch after $t_{k}$ as $b$.
Since the inter-arrival time is exponentially distributed and the inter-decision time is deterministically distributed, $b$ would be uniformly distributed over $[0,\frac{1}{\nu}]$.
    The AuD of decisions made during $Y_{k}$ conditioned on $X_{k}>T_{k-1}$, therefore, can be written as
\begin{align}
\Delta_\text{D}(\tau_{k_{i}})=T_{k-1}+b+\frac{j-1}{\nu},   j=1, 2, \cdot\cdot\cdot , N_{k}^{2}+N_{k}^{3}.\nonumber
\end{align}
Furthermore, the expected sum AuD during $Y_{k}$ would be
\begin{align}\label{eq:dayu}
\mathbb{E}&[\Delta_\text{Dk}^{2}|X_{k}\leq T_{k-1}]=\mathbb{E}[\Sigma_{i=1}^{N_{k}^{2}+N_{k}^{3}}\Delta_\text{D}(\tau_{k_{i}})|X_{k}>T_{k-1}]\nonumber \\
=&\mathbb{E}[T_{k-1}|X_{k}>T_{k-1}](\mathbb{E}[N_{k}^{2}]+\mathbb{E}[N_{k}^{3}])\nonumber \\
&+\frac{\mathbb{E}[(N_{k}^{2})^2]+\mathbb{E}[(N_{k}^{3})^2]+2(\mathbb{E}[N_{k}^{2}]+\mathbb{E}[N_{k}^{3}])}{2\nu}.
\end{align}

Finally, suppose that $K$ updates are served and $N_{\text{T}}$ decisions are made during a period $T$, where $K_{1}$ decisions are made during inter-departure times with $X_{k}\leq T_{k-1}$ and $K_{2}$ decisions are made during inter-departure times with $X_{k}>T_{k-1}$.
As $T$ goes to infinity, we have
\begin{align}\label{eq:V}
\lim_{T\rightarrow\infty}\frac{N_{\text{T}}}{K}=\frac{\mathbb{E}[Y_{k}]}{\frac{1}{\nu}}=\frac{\nu}{\lambda}.
\end{align}

Combing the results in \eqref{eq:xiaoyu}, \eqref{eq:dayu} and \eqref{eq:V}, the average AuD of M/G/1/D update-and-decide system can be expressed as
\begin{align}\label{eq:average}
&\widebar{\Delta}_\text{D}^\text{M/G/1/D}=\lim_{T\rightarrow\infty}\frac{1}{N_{\text{T}}} \sum_{k=1}^{K} \Delta_{\text{Dk}}\nonumber \\
&=\lim_{T\rightarrow\infty}\frac{K}{N_{\text{T}}}\big(\frac{K_{1}}{K}\frac{1}{K_{1}}\sum_{k=1}^{K_{1}}\Delta_{\text{Dk}}^{1}+\frac{K_{2}}{K}\frac{1}{K_{2}}\sum_{k=1}^{K_{2}}\Delta_{\text{Dk}}^{2}\Big)\nonumber \\
&=\frac{\lambda}{\nu}\Big(\rho\mathbb{E}[\Delta_\text{Dk}^{1}|X_{k}\leq T_{k-1}]+(1-\rho)\mathbb{E}[\Delta_\text{Dk}^{2}|X_{k}> T_{k-1}]\Big)\nonumber \\
&=\frac{\lambda\rho}{\nu}\Big(\mathbb{E}[T_{k-1}|X_{k}\leq T_{k-1}]\mathbb{E}[N_{k}^{1}]+\frac{\mathbb{E}[(N_{k}^{1})^{2}]}{2\nu}\Big)\nonumber \\
&~~~~+\frac{\lambda(1-\rho)}{\nu}\Big(\mathbb{E}[T_{k-1}|X_{k}>T_{k-1}](\mathbb{E}[N_{k}^{2}]+\mathbb{E}[N_{k}^{3}])\nonumber \\
&~~~~+\frac{\mathbb{E}[(N_{k}^{2})^2]+\mathbb{E}[(N_{k}^{3})^2]
+2(\mathbb{E}[N_{k}^{2}]+\mathbb{E}[N_{k}^{3}])}{2\nu}\Big),
\end{align}
which complete the proof of the \textit{Theorem} \ref{th:MG1D}.
\end{proof}

\subsection{Proof of \textit{Proposition} \ref{pro:tk}}\label{proof:tk}
\begin{proof}
As shown in \eqref{eq:tkxk}, we have $\Pr\{X_{k}\leq T_{k-1}\}=\rho$.
    Thus, the expectations of $T_{k}$ conditioned on $X_{k}\leq T_{k-1}$ and $X_{k}>T_{k-1}$ can be expressed as
\begin{align}
\mathbb{E}[T_{k-1}|X_{k}\leq T_{k-1}]=&\frac{1}{\rho}\int_{0}^{\infty}f_{\text{X}}(x)dx\int_{x}^{\infty}tf_{\text{T}}(t)dt\nonumber \\
=&\frac{1}{\rho}\int_{0}^{\infty}f_{\text{T}}(t)dt\int_{0}^{t}tf_{\text{X}}(x)dx\nonumber \\
=&\frac{1}{\rho}\int_{0}^{\infty}f_{\text{T}}(t)(t-te^{-\lambda t})dt\nonumber \\
=&\frac{\mathbb{E}[T_{k}]+w}{\rho}.
\end{align}
Likewise, we have
\begin{align}
\mathbb{E}[T_{k-1}|X_{k}>T_{k-1}]=&\frac{1}{1-\rho}\int_{0}^{\infty}f_{\text{X}}(x)dx\int_{0}^{x}tf_{\text{T}}(t)dt\nonumber \\
=&\frac{1}{1-\rho}\int_{0}^{\infty}f_{\text{T}}(t)dt\int_{t}^{\infty}tf_{\text{X}}(x)dx\nonumber \\
=&\frac{1}{1-\rho}\int_{0}^{\infty}te^{-\lambda t}f_{\text{T}}(t)dt\nonumber \\
=&-\frac{\omega}{1-\rho},
\end{align}
in which $\omega=\frac{\text{d}G_\text{T}(s)}{\text{d}s}|_{s=-\lambda}$.

This completes the proof of the \textit{Proposition} \ref{pro:tk}.
\end{proof}

\subsection{Proof of \textit{Corollary} \ref{cor:MG1D}}\label{proof:Average}
\begin{proof}
We denote $\omega_{0}=e^{-\frac{\mu}{\nu}}$ and $u_{0}=^{-\frac{\lambda}{\nu}}$.

By using Pollaczek-Khinchine formula \cite[Chap. 8, pp. 382]{Wolff-1989} for the FCFS M/G/1 queue, we have
\begin{align}
\mathbb{E}[T_{k}]=\frac{1}{\mu}+\frac{\lambda\mathbb{E}[S^{2}]}{2(1-\rho)}.
\end{align}

If $f_{\text{S}}(x)$ is a uniform distribution with parameter $\frac{2}{\mu}$, e.g., $f_{\text{SU}}(x)=\frac{\mu}{2}$, we have
\begin{align}
\mathbb{E}[T_{\text{U}k}]=\frac{1}{\mu}+\frac{2\rho}{3\mu(1-\rho)}.\nonumber
\end{align}

We denote the number of decisions made during $Y_{k}$ conditioned on $X_{k}\leq T_{k-1}$ as $N_{\text{U}k}^{1}$ and have
\begin{align}\label{eq:NUk1}
\Pr\{N_{\text{U}k}^{1}=0\}=&\Pr\{Y_{k}<a|X_{k}\leq T_{k-1}\}\nonumber \\
=&\int_{0}^{\frac{1}{\nu}}f_{\text{a}}(x)dx\int_{0}^{x}f_{\text{S}}(y)dy\nonumber \\
=&\frac{1}{4m_{0}},\nonumber \\
\Pr\{N_{\text{U}k}^{1}=j\}=&\Pr\left\{\frac{j-1}{\nu}+a<Y_{k}<\frac{j}{\nu}+a|X_{k}\leq T_{k-1}\right\}\nonumber \\
=&\int_{0}^{\frac{1}{\nu}}f_{\text{a}}(x)dx\int_{\frac{j-1}{\nu}+x}^{\frac{j}{\nu}+x}f_{\text{S}}(y)dy\nonumber \\
=&\frac{1}{2m_{0}}, j=1, 2, \cdot\cdot\cdot ,\nonumber\\
\mathbb{E}[N_{\text{U}k}^{1}]=&\sum_{j=1}^{2m_{0}}j\Pr\{N_{\text{U}k}^{1}=j\}=\frac{2m_{0}+1}{2},\nonumber \\
\mathbb{E}[(N_{\text{U}k}^{1})^{2}]=&\frac{(2m_{0}+1)(2m_{0}+2)(4m_{0}+3)}{6}.
\end{align}

Since inter-arrival time $X_{k}$ is exponentially distributed, we have $\Pr\{X-T>y|X>T\}=\Pr\{X>y\}$, i.e.,
\begin{align}
f_{\text{X-T}}(x)=\lambda e^{-\lambda }.\nonumber
\end{align}

We denote the number of decisions made during $Y_{k1}$ and $Y_{k2}$ as $N_{\text{U}k}^{2}$ and $N_{\text{U}k}^{3}$ and have
\begin{align}\label{eq:NUk2}
\Pr\{N_{\text{U}k}^{2}=0\}&=\Pr\{Y_{k1}<b|X_{k}> T_{k-1}\}\nonumber \\
&=\Pr\{X_{k}-T_{k-1}<b|X_{k}> T_{k-1}\}\nonumber \\
&=\int_{0}^{\frac{1}{\nu}}f_{\text{b}}(x)dx\int_{0}^{x}f_{\text{X-T}}(y)dy\nonumber \\
&=1-\frac{\nu(1-u_{0})}{\lambda}\triangleq p_\text{Us},\nonumber \\
\Pr\{N_{\text{U}k}^{2}=j\}&=\Pr\Big\{\frac{j-1}{\nu}+b<Y_{k}<\frac{j}{\nu}+b|X_{k}> T_{k-1}\Big\}\nonumber \\
&=\int_{0}^{\frac{1}{\nu}}f_{\text{b}}(x)dx\int_{\frac{j-1}{\nu}+x}^{\frac{j}{\nu}+x}f_{\text{X-T}}(y)dy\nonumber \\
&=(1-p_\text{Us})(1-u_{0})u_{0}^{j-1}, j=1, 2, \cdot\cdot\cdot ,\nonumber \\
\mathbb{E}[N_{\text{U}k}^{2}]=&\frac{\nu}{\lambda},\mathbb{E}[(N_{\text{U}k}^{2})^{2}]
=\frac{\nu(1+u_{0})}{\lambda(1-u_{0})}.
\end{align}

We denote the length between departure epoch $t_{k-1}^{'}$ and the next decision epoch after $t_{k-1}^{'}$ as $c$.
    By approximating the distribution of $c$ with a uniform distribution with parameter $1/\nu$, i.e.,
\begin{align}
f_{\text{c}}(x)=\nu, x\in\Big(0,\frac{1}{\nu}\Big).\nonumber
\end{align}

We then have
\begin{align}
\Pr\{N_{\text{U}k}^{3}=0\}&=\Pr\{Y_{k2}<c|X_{k}> T_{k-1}\}\nonumber \\
&=\Pr\{S_{k}<c|X_{k}> T_{k-1}\}\nonumber \\
&=\int_{0}^{\frac{1}{\nu}}f_{\text{c}}(x)dx\int_{0}^{x}f_{\text{S}}(y)dy\nonumber \\
&=\frac{2m_{0}-1}{2},\nonumber \\
\Pr\{N_{\text{U}k}^{3}=j\}&=\Pr\Big\{\frac{j-1}{\nu}+c<Y_{k}<\frac{j}{\nu}+c|X_{k}> T_{k-1}\Big\}\nonumber \\
&=\int_{0}^{\frac{1}{\nu}}f_{\text{c}}(x)dx\int_{\frac{j-1}{\nu}+x}^{\frac{j}{\nu}+x}f_{\text{S}}(y)dy\nonumber \\
&=\frac{1}{2m_{0}}, j=1, 2, \cdot\cdot\cdot ,\nonumber
\end{align}
\begin{align}\label{eq:NUk3}
\mathbb{E}[N_{\text{U}k}^{3}]=&\mathbb{E}[N_{\text{U}k}^{1}]=\frac{2m_{0}+1}{2},\nonumber \\
\mathbb{E}[(N_{\text{U}k}^{3})^{2}]=&\frac{(2m_{0}+1)(2m_{0}+2)(4m_{0}+3)}{12m_{0}}.
\end{align}

By combining \eqref{eq:average}, \eqref{eq:NUk1}, \eqref{eq:NUk2} and \eqref{eq:NUk3}, the average AuD of an M/U/1/D update-and-decide system can be obtained as
\begin{align}
\widebar{\Delta}_\text{D}^\text{M/U/1/D}=&\frac{2e^{2\rho}(1-\rho)}{\mu(e^{2\rho}-1)}+\frac{(1-\rho)(1+u_{0})}{2\mu m_{0}(1-u_{0})}-\frac{1}{\mu\rho(1-\rho)}\nonumber \\
&-\frac{(2m_{0}^{2}+1)\rho^{2}+(16m_{0}^2-1)\rho-36m_{0}^{2}-6m_{0}}
{12m_{0}^{2}\mu(1-\rho)}.\nonumber
\end{align}

If $f_{\text{S}}(x)$ is an exponential distribution with parameter $\mu$, e.g, $f_{\text{SE}}(x)=\mu e^{-\mu x}$, we have
\begin{align}
\mathbb{E}[T_{\text{E}k}]=\frac{1}{\mu}+\frac{\rho}{\mu(1-\rho)}.
\end{align}

We denote the number of decisions made during $Y_{k}$ conditioned on $X_{k}\leq T_{k-1}$ as $N_{\text{E}k}^{1}$ and have
\begin{align}
\Pr\{N_{\text{E}k}^{1}=0\}&=\Pr\{Y_{k}<a|X_{k}\leq T_{k-1}\}\nonumber \\
&=\int_{0}^{\frac{1}{\nu}}f_{\text{a}}(x)dx\int_{0}^{x}f_{\text{S}}(y)dy\nonumber \\
&=1-\frac{\nu(1-\omega_{0})}{\mu}\triangleq p_\text{Es},\nonumber \\
\Pr\{N_{\text{E}k}^{1}=j\}&=\Pr\Big\{\frac{j-1}{\nu}+a<Y_{k}<\frac{j}{\nu}+a|X_{k}\leq T_{k-1}\Big\}\nonumber \\
&=\int_{0}^{\frac{1}{\nu}}f_{\text{a}}(x)dx\int_{\frac{j-1}{\nu}+x}^{\frac{j}{\nu}+x}f_{\text{S}}(y)dy\nonumber \\
&=(1-p_\text{Es})(1-\omega_{0})\omega_{0}^{j-1}, j=1, 2, \cdot\cdot\cdot ,\nonumber
\end{align}
\begin{align}\label{eq:NEk1}
\mathbb{E}[N_{\text{E}k}^{1}]=\frac{\nu}{\mu}, \mathbb{E}[(N_{\text{E}k}^{3})^2]=\frac{\nu(1+\omega_{0})}{\mu(1-\omega_{0})}.
\end{align}

We denote the number of decisions made during $Y_{k1}$ and $Y_{k2}$ as $N_{\text{E}k}^{2}$ and $N_{\text{E}k}^{3}$. We have
\begin{align}\label{eq:NEk2k3}
&\mathbb{E}[N_{\text{E}k}^{2}]=\mathbb{E}[N_{\text{U}k}^{2}]=\frac{\nu}{\lambda}, \mathbb{E}[(N_{\text{E}k}^{2})^2]=\frac{\nu(1+u_{0})}{\lambda(1-u_{0})},   \nonumber \\
&\mathbb{E}[N_{\text{E}k}^{3}]=\mathbb{E}[N_{\text{E}k}^{1}]=\frac{\nu}{\mu}, \mathbb{E}[(N_{\text{E}k}^{3})^2]=\frac{\nu(1+\omega_{0})}{\mu(1-\omega_{0})}.
\end{align}

By combining \eqref{eq:average}, \eqref{eq:NEk1} and \eqref{eq:NEk2k3}, the average AuD of an M/M/1/D update-and-decide system can be obtained as
\begin{align}
\widebar{\Delta}_\text{D}^\text{M/M/1/D}=\frac{2\rho^2-3\rho+2}{\mu(1-\rho)}+\frac{\rho(1+\omega_{0})}{2\mu m_{0}(1-\omega_{0})}+\frac{(1-\rho)(1+u_{0})}{2\mu m_{0}(1-u_{0})}.\nonumber
\end{align}
If $f_{\text{S}}(x)$ is a deterministic distribution with parameter $\frac{1}{\mu}$, e.g, $f_{\text{SD}}(x)=\delta\left(x-\frac{1}{\mu}\right)$, we have
\begin{align}
\mathbb{E}[T_{\text{D}k}]=\frac{1}{\mu}+\frac{\rho}{2\mu(1-\rho)}.
\end{align}

We further denote the number of decisions made during $Y_{k}$ conditioned on $X_{k}\leq T_{k-1}$ as $N_{\text{D}k}^{1}$.
    Since the service time is $S_{k}=\frac{1}{\mu}$, the inter-decision time is $Z_{j}=\frac{1}{\nu}$ in which $\nu=m_{0}\mu$, we have $N_{\text{D}k}^{1}=m_{0}$.

We also denote the number of decisions made during $Y_{k1}$ and $Y_{k2}$ as $N_{\text{E}k}^{2}$ and $N_{\text{E}k}^{3}$. We have
\begin{align}\label{eq:NDk123}
\mathbb{E}[N_{\text{D}k}^{2}]&=\mathbb{E}[N_{\text{U}k}^{2}]=\frac{\nu}{\lambda}, \mathbb{E}[(N_{\text{E}k}^{2})^2]=\frac{\nu(1+u_{0})}{\lambda(1-u_{0})}, \nonumber \\
N_{\text{D}k}^{3}&=N_{\text{D}k}^{1}=m_{0}.
\end{align}

By combining \eqref{eq:average} and \eqref{eq:NDk123}, the average AuD of an M/D/1/D update-and-decide system can be obtained as
\begin{align}
\widebar{\Delta}_\text{D}^\text{M/D/1/D}=&\frac{m_{0}e^{\rho}(1-\rho)}{\mu\rho^{2}}+\frac{(1-\rho)(1+u_{0})}{2\mu\rho(1-u_{0})} \nonumber \\
&+\frac{-3m_{0}+6m_{0}\rho-2m_{0}}{2\mu\rho^{2}(1-\rho)}.\nonumber
\end{align}

This completes the proof of \textit{Corollary} \ref{cor:MG1D}.
\end{proof}

\subsection{Proof of \textit{Proposition} \ref{prop:MG1Dbijiao}}\label{proof:MG1Dbijiao}
\begin{proof}
Note that the average AuDs of the M/U/1/D, M/M/1/D and M/D/1/D systems have been obtained in \textit{Corollary} \ref{cor:MG1D}.

Be defining the following auxiliary function
\begin{align}
f_{1}(m_{0})=&\widebar{\Delta}_\text{D}^\text{M/M/1/D}-\widebar{\Delta}_\text{D}^\text{M/U/1/D} \nonumber \\
=&\frac{\rho(1+\omega_{0})}{2\mu m_{0}(1-\omega_{0})}-\frac{1}{2m_{0}\mu(1-\rho)}-\frac{\rho}{12m_{0}^{2}\mu}  \nonumber \\
&-\frac{2e^{2\rho}(1-\rho)}{\mu(e^{2\rho}-1)}+\frac{13\rho^{3}-10\rho^{2}-6\rho+6}{6\mu\rho(1-\rho)}
\end{align}
and set $m_{0}=-\frac{1}{t}$, $-1\leq t<0$, the first order derivative $f'_1(m_{0})$ can be given by
\begin{align}
f_{1}^{'}(t)=&f_{1}^{'}\left(-\frac{1}{m_{0}}\right) \nonumber \\
\label{apx:h_1}
=&\frac{\rho^2-\rho+1}{2\mu(1-\rho)} + \frac{\rho e^{t}(e^{t}-t-1)}{\mu(1-e^{t})^2}-\frac{\rho t}{6\mu}.
\end{align}

Since we have $0<\rho<1$ and $-1\leq t<0$, it can be easily seen that each item of \eqref{apx:h_1} is positive, and thus $f_{1}^{'}(t)>0$.
    Since we also have $f_{1}(1)<0$ and $f_{1}(\infty)>0$, it is concluded that we can find such an $m_{0}^{*}$ that $f_{1}(m_{0}^{*}-1)\leq 0$ and $f_{1}(m_{0}^{*})>0$.
    That is, there exists an $m_0^*$ satisfying $\widebar{\Delta}_\text{D}^\text{M/M/1/D}< \widebar{\Delta}_\text{D}^\text{M/U/1/D}$ if  $m_{0}< m_{0}^{*}$ and $\widebar{\Delta}_\text{D}^\text{M/M/1/D}> \widebar{\Delta}_\text{D}^\text{M/U/1/D}$ if  $m_{0}\geq m_{0}^{*}$.

Likewise, by checking the derivative of 
\begin{align}
f_{2}(m_{0})=&\widebar{\Delta}_\text{D}^\text{M/U/1/D}-\widebar{\Delta}_\text{D}^\text{M/D/1/D} \nonumber \\
=&\frac{e^{\rho}(1-\rho)(2\rho e^{\rho}-e^{2\rho}+1)}{\mu\rho(e^{2\rho}-1)} \nonumber \\
&+\frac{\rho(2m_{0}^{2}+1)}{12m_{0}^{2}\mu}+\frac{1}{2m_{0}\mu(1-\rho)},
\end{align}
it can be seen that $f_{2}(m_{0})$ is decreasing with $m_{0}$ (since $m_0$ is in the denominators of the three items) and $f_{2}(\infty)>0$.
Thus, we have $\widebar{\Delta}_\text{D}^\text{M/U/1/D}>\widebar{\Delta}_\text{D}^\text{M/D/1/D}$ for all $m_0\geq1$.
Thus, the proof of \textit{Proposition} \ref{prop:MG1Dbijiao} is completed.

\end{proof}

\subsection{Proof of \textit{Proposition} \ref{prop:MG1Dpmis}}\label{proof:MG1DPmis}
\begin{proof}
The event that an update is missed to make any decision is equivalent to the event that the inter-departure time before the update is less than an inter-decision time.
As shown in \eqref{eq:yk}, we have $Y_{k}=S_{k}$ or $Y_{k}=X_{k}-T_{k-1}+S_{k}$.
However, when the service time is uniformly distributed or exponentially distributed, the update-and-decide system is randomly served.
Thus, the inter-departure time is randomly distributed.
\begin{figure}[!t]
  \centering
  \includegraphics[width=2.0in]{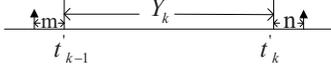}\\
  \caption{inter-departure time}
  \label{fig:deterpmis}
\end{figure}
We set $m$ to represent the length between departure epoch $t_{k-1}^{'}$ and the last decision epoch before $t_{k-1}^{'}$ and set $n$ to represent the length between departure epoch $t_{k}^{'}$ and the next decision epoch after $t_{k}^{'}$.
Approximate that $d=m+n$ is uniformly distributed with parameter $\frac{2}{\nu}$, i.e.,
\begin{align}
f_{\text{d}}(x)=\frac{\nu}{2}, x\in\Big(0,\frac{2}{\nu}\Big).\nonumber
\end{align}
Approximately, we have
\begin{align}
p_{\text{mis}}=&\Pr\{Y_{k}+d<\frac{1}{\nu}\}\nonumber \\
=&\Pr\{X_{k}\leq T_{k-1}\}\Pr\{Y_{k}<\frac{1}{\nu}|X_{k}\leq T_{k-1}\}\nonumber \\
&+\Pr\{X_{k}>T_{k-1}\}\Pr\{Y_{k}<\frac{1}{\nu}|X_{k}>T_{k-1}\}  \nonumber \\
=&\rho\Pr\{S_{k}+d<\frac{1}{\nu}\}  \nonumber \\
&+(1-\rho)\Pr\{X_{k}-T_{k-1}+S_{k}+d<\frac{1}{\nu}\}\nonumber \\
=&\rho\int_{0}^{\frac{1}{\nu}}f_{\text{S}}(x)dx\int_{0}^{\frac{1}{\nu}-x}f_{\text{d}}(z)dz+(1-\rho)\nonumber \\
&\cdot\int_{0}^{\frac{1}{\nu}}f_{\text{S}}(x)dx\int_{0}^{\frac{1}{\nu}-x}f_{\text{X-T}}(t)dt\int_{0}^{\frac{1}{\nu}-x-t}f_{\text{d}}(z)dz\nonumber \\
=&\int_{0}^{\frac{1}{\nu}}\frac{1}{2}f_{\text{S}}(x)(1-\nu x)dx \nonumber \\
&+(1-\rho)\int_{0}^{\frac{1}{\nu}}\frac{1}{2}f_{\text{S}}(x)(\frac{\nu u_{0}e^{\lambda x}}{\lambda}-\frac{\nu}{\lambda}),\nonumber \\
p_{\text{mis}}^{\text{M/U/1/D}}=&\frac{1}{8m_{0}}+\frac{(1-\rho)(m_{0}(1-u_{0})-\rho)}{4\rho^{2}},\nonumber \\
p_{\text{mis}}^{\text{M/M/1/D}}=&\frac{1}{2}+\frac{m_{0}(u_{0}-1)}{2\rho},
\end{align}
where, $u_{0}=e^{-\frac{\lambda}{\nu}}$.

In M/D/1/D update-and-decide systems, the inter-departure time $Y_{k}$ will consist no less than $m_{0}$ decision epochs.
Due to $m_{0}\geq1$, the missing probability $p_{\text{mis}}^{\text{M/D/1/D}}$ will always equal to 0, i.e., $p_{\text{mis}}^{\text{M/D/1/D}}=0$.

This completes the proof of the \textit{Proposition} \ref{prop:MG1Dpmis}.
\end{proof}

\end{document}